\documentclass[12pt]{article}
\usepackage{fullpage}
\usepackage{complexity}
\usepackage{amsmath}
\usepackage{amsthm}
\usepackage{graphicx}

\def\ensurenomath#1{\relax\ifmmode \mbox{#1}\else#1\fi}

\newcommand{\lreach}[1]
           {\ensuremath{\operatorname{\lang{#1}--\lang{\ensurenomath{\sc Reach}}}}}
\newcommand{\ldreach}[1]{\ensuremath{\operatorname{\lang{#1}--\lang{{\ensurenomath{\sc
          DagReach}}}}}}
\newcommand{\ltreach}[1]{\ensuremath{\operatorname{\lang{#1}--\lang{{\ensurenomath{{\sc
          TreeReach}}}}}}}
\newcommand{\lureach}[1]{\ensuremath{\operatorname{\lang{#1}--\lang{\ensurenomath{\sc
          UReach}}}}}
\newcommand{\dr}[1]{\ensurenomath{{\sc Reach}}\ensuremath{^{-1}(\class{#1})}}
\newcommand{\ur}[1]{\ensurenomath{{\sc UReach}}\ensuremath{^{-1}(\class{#1})}}
\newcommand{\dagr}[1]{\ensurenomath{{\sc DagReach}}\ensuremath{^{-1}(\class{#1})}}
\newcommand{\treer}[1]{\ensurenomath{{\sc TreeReach}}\ensuremath{^{-1}(\class{#1})}}

\newcommand{\creach}[1]{\ensuremath{\operatorname{\class{#1}--\lang{{\ensurenomath{{\sc
          Reach}}}}}}}
\newcommand{\ctreach}[1]{\ensuremath{\operatorname{\class{#1}--\lang{{\ensurenomath{\sc
        TreeReach}}}}}}
\newcommand{\cdreach}[1]{\ensuremath{\operatorname{\class{#1}--\lang{{\ensurenomath{\sc
        DagReach}}}}}}
\newcommand{\cureach}[1]{\ensuremath{\operatorname{\class{#1}--\lang{{\ensurenomath{\sc
        UReach}}}}}}

\newcommand{\LogCFL}{\class{LogCFL}}
\newcommand{\mach}[1]{\class{#1}}

\newcommand{\nauxpdapoly}{\mach{AuxPDA(poly)}}

\newcommand{\lmreduces}{\ensuremath{\leq_{m}^{\L}}}

\newtheorem{observation}{Observation}
\newtheorem{definition}{Definition}
\newtheorem{proposition}{Proposition}
\newtheorem{theorem}{Theorem}

\usepackage{color}

\newenvironment{proposition-a}[1]{\noindent {\bf Proposition~#1:~}\em }{\smallskip}
\newenvironment{lemma-a}[1]{\noindent {\bf Lemma~#1:~}\em }{\smallskip}
\newenvironment{theorem-a}[1]{\noindent {\bf Theorem~#1:~}\em
}{\smallskip}

\author{Balagopal Komarath$^*$ \and Jayalal Sarma$^*$ \and K. S. Sunil
  \thanks{Department of Computer Science \&
    Engineering, Indian Institute of Technology Madras, Chennai -- 36,
    India. Email : {\{\tt baluks,jayalal,sunil\}@cse.iitm.ac.in} 
    The first author was supported by the TCS Ph.D. Fellowship.}
}

\title{On the Complexity of {\sc L}-reachability\footnote{A preliminary version of this work with a subset of results was presented at {\em 16th International Workshop on Descriptional Complexity of Formal Systems (DCFS 2014)} and appears in \cite{KomarathSS14}.}}
\begin{document}
\maketitle

\begin{abstract}
  We initiate a complexity theoretic study of the language based
  graph reachability problem~(\lreach{\lang{L}}) : Fix a language
  \lang{L}. Given a graph whose edges are labelled with alphabet
  symbols of the language \lang{L} and two special vertices $s$ and $t$, test if there is path
  $P$ from $s$ to $t$ in the graph such that the concatenation of the
  symbols seen from $s$ to $t$ in the path $P$ forms a string in the
  language \lang{L}. We study variants of this problem with different
  graph classes and different language classes and obtain complexity
  theoretic characterizations for all of them. Our main results are
  the following:

  \begin{itemize}
    \item Restricting the language using formal language theory we
      show that the complexity of \lreach{\lang{L}} increases with the
      power of the formal language class. We show that there is a
      regular language for which the \lreach{\lang{L}} is
      $\NL$-complete even for undirected graphs.  In the case of
      linear languages, the complexity of \lreach{\lang{L}} does not go
      beyond the complexity of \lang{L} itself. Further, there is a
      deterministic context-free language \lang{L} for which \ldreach{\lang{L}}  
      is \LogCFL-complete.
      
    \item We use  \lreach{\lang{L}}  
      as a lens to study structural
      complexity. In this direction we show that there is a language
      $\lang{A}$ in \TC$^0$ for which $\ldreach{A}$ is
      $\NP$-complete. Using this we show that $\P$ vs $\NP$ question is
      equivalent to $\P$ vs $\dagr{\P}$\footnote{For any complexity class 
      \class{C}, $\dagr{\class{C}} = \{ \lang{L} : \ldreach{L} \in \class{C}\}$.} 
      question. This leads to the
      intriguing possibility that by proving \dagr{\P} is contained in
      some subclass of \P, we can prove an upward translation of
      separation of complexity classes. Note that we do not know a way
      to upward translate the separation of complexity classes.
  \end{itemize}
\end{abstract}

%\begin{keywords}
%\lang{L}-reachability, \class{CFL}-reachability, structural complexity
%\end{keywords}

\section{Introduction}

Reachability problems in mathematical structures are a well-studied
problem in space complexity. An important example is the graph
reachability problem where given a directed graph $G$ and two special vertices
$s$ and $t$, is there a path\footnote{We follow the convention that a path can have repeated vertices and edges.} from $s$ to $t$ in the graph $G$. This
problem exactly captures the space complexity of problems solvable in
nondeterministic logarithmic space. Various restrictions of the
problem have been studied - reachability in undirected graphs
characterizes deterministic logspace \cite{Reingold05}, reachability in
constant width graphs (even undirected) characterizes $\NC^1$ \cite{Barrington89}, reachability in
planar constant width directed graphs characterizes $\ACC^0$ \cite{Hansen04} and the
version in upward planar constant width directed graphs characterizes $\AC^0$
\cite{Barrington98searchingconstant}.

A natural extension of the problem using formal language theory is the
\lreach{\lang{L}} problem: Fix a language \lang{L} defined over a
finite alphabet $\Sigma$. Given a graph whose edges are labelled by
alphabet symbols and two special vertices $s$ and $t$, test if there
is path from $s$ to $t$ in the graph such that the concatenation of
the symbols seen from $s$ to $t$ forms a string in the language
\lang{L}. Indeed, if \lang{L} is $\Sigma^*$, then the string on any
path from $s$ to $t$ will be in the language. Hence the problem
reduces to the graph reachability problem.

Although \lreach{\lang{L}} problem has not been studied from a space complexity theory
perspective, a lot is known about its complexity~\cite{Reps96,Reps98,Yanna90,Reps90,Marathe00}.  An immediate observation is that the
\lreach{\lang{L}} problem is at least as hard as the membership problem
of \lang{L}. Indeed, given a string $x$, to check for membership in \lang{L} it
suffices to test \lreach{\lang{L}} in a simple path of length $|x|$
where the edges are labelled by the symbols in $x$ in that sequence. The
literature on the problem is spread over two main themes. One is on
restricting the language from the formal language perspective, and the
other is by restricting the family of graphs in terms of structure.

An important special case of the problem that was studied is when the
language is restricted to be a context-free language (CFL). This is called
the \lreach{\class{CFL}}. A primary motivation to study this
problem is their application in various practical situations like
inter-procedural slicing and inter-procedural data flow
analysis~\cite{Reps90,Reps96,Reps98}.  These are used in code
optimization, vectorization and parallelization phases of compiler
design where one should have information about reaching definitions,
available expressions, live variables, etc. associated with the
program elements. The goal of inter-procedural analysis is to perform
static examination of above properties of a program that consists of
multiple procedures. Once a program is represented by its program
dependence graph ~\cite{Reps96}, the slicing problem is simply the
\lreach{\class{CFL}} problem.

\paragraph{\bf Our Results: } The results in this paper are in two flavors. \\
\noindent
{\bf Results based on Chomsky Hierarchy and Graph Classes: } Firstly we study restrictions of \lreach{\lang{L}} problem when
$\lang{L}$ is restricted using formal language hierarchy and the graph is
restricted to various natural graph classes. Our results on this front are listed in Table~\ref{table:FL_Reach} (for the sake of completeness, we include
some known results too). Apart from the results in Table~\ref{table:FL_Reach}, we show the following theorem for the language class \DCFL.
\begin{theorem}
\label{thm:main}
 \cdreach{DCFL} is $\LogCFL$-complete.
\end{theorem}
%\vspace{-5mm}
\begin{table}
\label{table:FL_Reach}
\begin{minipage}{\textwidth}
\centering
\caption{Formal language class restricted reachability} %{\small
%    ($^\dagger$ indicates easy corollary of known results)}}
\label{table2}
\begin{tabular}{|l|l|l|l| }
\hline
\textbf{Language Class}& \sc{Tree-Reach} & \sc{DAG-Reach} & \sc{UReach/Reach} \\
\hline 
Regular & \L-complete\cite{Yanna90} & \NL-complete\cite{Yanna90} & \NL-complete \\
& & & (Theorem~\ref{thm:regureach}/\cite{Yanna90}) \\
\hline 
Linear & \NL-complete & \NL-complete & \NL-complete \\
 & (Theorem~\ref{thm:lin}) & (Theorem~\ref{thm:lin}) & (Theorem~\ref{thm:lin}) \\
\hline
Context-free & \LogCFL-complete & \LogCFL-complete  & \P-complete \\
 & (Prop.~\ref{prop:cfl-tdreach}) & (Prop.~\ref{prop:cfl-tdreach}) & (Theorem~\ref{thm:cfl-reach-pc}/\cite{UllmanG88}) \\
\hline
Context-sensitive & \PSPACE-complete & \PSPACE-complete & Undecidable \\
 & (Prop.~\ref{prop:csl-tdreach}) & (Prop.~\ref{prop:csl-tdreach}) & (\cite{Marathe00}) \\
\hline
\end{tabular}
\end{minipage}
\end{table}
%\vspace{-5mm}

\vspace{.1cm}
\noindent 
{\bf Results on the Structural Complexity front:} Now we take a complexity theoretic view, where we study
$\lreach{\lang{L}}$ as an operator on languages. It is shown in Barrett
et. al. \cite{Marathe00} that even for languages in logspace, the languages
\lreach{L} and \lureach{L} are undecidable. Therefore in this section,
we consider only DAGs. Note that for any language \lang{L}, the language
\ldreach{L} is decidable.

It is natural to ask whether increasing the complexity of \lang{L}
increases the complexity of \ldreach{L}. More concretely, does
$\lang{A}$ $\lmreduces$ $\lang{B}$ $\implies$ $\ldreach{A}$ $\lmreduces$
$\ldreach{B}$? The following theorem, along with the fact that there
exists a language \lang{L} (see Proposition~\ref{prop:cfl-tdreach}) that is $\LogCFL$-complete for which
\ldreach{L} remains $\LogCFL$-complete shows that such a result is
highly unlikely.

\begin{theorem}
\label{thm:main}
  There exists a language $\lang{A} \in \TC^0$ for which $\ldreach{A}$ is \NP-complete.
\end{theorem}

\noindent For any complexity class $\class{C}$, we consider the class of
languages defined as,
\[ \dagr{\class{C}} = \{ \lang{L} : \ldreach{L} \in \class{C}\} \]
Note that for any class $\class{C}$, we have
$\dagr{\class{C}} \subseteq \class{C}$.
We have the following theorems for different choices of 
$\class{C}$.\\

\begin{theorem-a}{}
  We show the following structural theorems:
  \begin{enumerate}
  \item (Theorem~\ref{thm:pspace-np-inv}) $\dagr{\PSPACE} = \PSPACE$, $\dagr{\NP} = \NP$.
  \item (Theorem~\ref{thm:pdr}) $\P \neq \dagr{\P} \iff \P \neq \NP$.
  \item (Theorem~\ref{thm:np-nl}) $\dagr{\NL} \neq \NL \iff \NP \neq \NL$.
  \end{enumerate}
  \label{thm:struc}
\end{theorem-a}

The above theorem shows that separating $\dagr{\P}$ from $\P$ would
separate $\P$ from $\NP$.  This gives us an upward translation of
lower bounds on complexity classes if we can prove that \dagr{\P}
is contained in some subclass of \P.  Hence the question whether we
can identify some ``natural'' complexity class containing $\dagr{\P}$
becomes very interesting.  It is clear that $\dagr{\P}$ contains
\LogCFL-complete problems but is highly unlikely to contain some
problems in $\L$. If \dagr{\P} contains some \P-complete problem, then
proving that \dagr{\P} is contained in some subclass of \P\ would be
very hard. In this connection, we show the following: 

\begin{theorem}
  If $\lang{L}$ is $\P$-complete under syntactic read-once logspace
  reductions, then $\ldreach{L}$ is $\NP$-complete.
  \label{thm:rol}
\end{theorem}

If we are able to extend the above theorem to all types of reductions, then it implies that, assuming $\NP$ is not contained in $\P$, $\dagr{\P}$ is unlikely to contain $\P$-complete problems. In other words, the above theorem could be interpreted as evidence (albeit very weak evidence) that \dagr{\P} may indeed be contained in some subclass of
\P. 

We also remark that Theorem~\ref{thm:pdr} holds with
\NL\ instead of \P. However, since \dagr{\NL} contains \NL-complete (under logspace
reductions) languages, Theorem~\ref{thm:np-nl} is not as promising (as Theorem~\ref{thm:pdr}).

A preliminary version of this with a subset of results appears \cite{KomarathSS14}. In \cite{KomarathSS14}, we proved that there is a language $A$ in {\em logspace} such that $\ldreach{A}$ is $\NP$-complete. In this extended version, we improve this bound to $\TC^0$ (from logspace, see Theorem~\ref{thm:main}).

\section{Preliminaries}
In this section, we define language restricted reachability problems
and make some observations on their complexity. The definitions for
standard complexity classes and their complete problems that we are
using in this paper can be found in standard complexity theory
textbooks \cite{arora-barak}. We use \L\ and \NL\ to stand for the
complexity classes deterministic logspace and nondeterministic
logspace respectively. All reductions (even ones used for defining
completeness) in this paper are in logspace unless mentioned
otherwise.

\begin{definition}

For any language $\lang{L} \subseteq \Sigma^{*}$, we consider graph $G$ where each
edge in $G$ is labelled by an element from $\Sigma$. For any path in
$G$ we define the yield of the path as the string formed by
concatenating the symbols found in the path in that order. Then we
define the language \lreach{L} as the set of all $(G, s, t)$ such that
there exists a path from $s$ to $t$ in $G$ with yield in
$\lang{L}$.

\label{def:lreach}
\end{definition}

By restricting the graph in Definition~\ref{def:lreach}, we obtain
similar definitions for \ldreach{L} (DAGs), \lureach{L} (Undirected
Graphs) and \ltreach{L} (Orientations of Undirected Trees).

%\begin{definition}
Let $\Sigma$ and $\Gamma$ be finite alphabets.  A function $f$ from
$\Sigma^*$ to $\Gamma^*$ is called a projection if for all $x \in
\Sigma^*$, the string $f(x) = y$ is such that for all $i \in [m]$,
either $y_i = x_j$ for some $j \in [n]$ or $y_i = 0$ or $y_i = 1$,
where $m = |y|$ and $n = |x|$. A language \lang{L} over $\Sigma$ is
said to be projection reducible to a language \lang{L'} over $\Gamma$ if there is
a projection $f$ such that $x \in L \iff y = f(x) \in$ \lang{L'}
and $|y|$ is polynomial in $|x|$.
%\label{def:projection}
%\end{definition}

\begin{observation}
 Any language \lang{L} is projection reducible to \ltreach{L}.
\label{obs:L-to-Lreach}
\end{observation}

Clearly, the above observation holds for any reachability variant
based on the graph. This is because \ltreach{L} is a restriction of the other reachability variants. In fact the following observation shows that
\ltreach{L} is not much harder than \lang{L}.

\begin{observation}
For any language \lang{L}, the language \ltreach{L} is logspace
reducible to \lang{L}.
\label{obs:Ltreach-to-L}
\end{observation}

Observation~\ref{obs:Ltreach-to-L} holds because in logspace we can
find the unique path (and hence its yield) from $s$ to $t$ in some tree and run the
algorithm for \lang{L} on the yield.

\vspace{.1cm}
Next we define classes of languages based on language restricted
reachability. 

\begin{definition}
For any class of languages $\class{C}$, we define the set of languages \creach{C} as the class of all languages \lreach{L} where $L$ is in $\class{C}$.
\label{def:creach}
\end{definition}

\noindent Again, by restricting graphs in Definition~\ref{def:creach}, we obtain
similar definitions for \cdreach{C}, \cureach{C} and \ctreach{C}.

\begin{definition}
For a class of languages \class{C} and a complexity class \class{D}, 
we say that \emph{\creach{C} is complete for \class{D}} if the following conditions 
are satisfied.

\begin{itemize}
\item For all $\lang{L} \in \class{C}$, the language \lreach{L} is in \class{D}. 
\item There exists a language \lang{L} in \class{C} such that the language \lreach{L}
  is hard for \class{D}.
\end{itemize}
\end{definition}

\begin{definition}
For any complexity class \class{C}, we define \dr{C} as the set of all
languages \lang{L} such that \lreach{L} is in \class{C}.
\label{def:rinv}
\end{definition}

Again, by restricting graphs in Definition~\ref{def:rinv}, we obtain
similar definitions for \dagr{C}, \ur{C} and \treer{C}.

Note that by Observation~\ref{obs:L-to-Lreach}, for any class
\class{C} the relations $\dr{C} \subseteq$ $\dagr{C} \subseteq$ $\treer{C}
\subseteq \class{C}$ holds. In this paper, we will be mainly studying
\dagr{C} for many interesting complexity classes \class{C}.

Our motivation in studying \dagr{C} is that it seems that it may be
helpful in proving upward translation of separation of complexity
classes. Note that we already know, by a standard padding argument,
how to translate separations of complexity classes downwards. For
example, we know that $\NEXP \neq \EXP \implies \P \neq \NP$. The
central question that we address is the following - For a class
\class{C}, what is the complexity of \dagr{C}? Clearly \dagr{C} is
contained in \class{C}. But for many natural complexity classes
\class{C}, \class{D} and \class{E}, if we can show that if \dagr{C} is
contained in some subclass \class{D} of \class{C}, then separating
\class{C} and \class{D} is equivalent to separating \class{C} from
some complexity class \class{E} that contains \class{C}.

We use \class{REG}, \class{CFL} and \class{CSL} to stand for
well-known formal language classes of regular, context-free and
context-sensitive languages respectively\cite{hopcroft-ullman}. The
formal language class \class{LIN}, called the set of all linear
languages, is the set of all languages with a context-free grammar
where the right-hand side of each production consists of at most one
non-terminal. The class \class{LIN} can also be characterized as
\class{CFL}s that can be decided by 1-turn PDAs (sub-family of PDAs
where for any computation, the stack height switches only once from 
non-decreasing mode to non-increasing mode).

We now state a known result with its proof idea which will be used
later in the paper.

\begin{proposition}[\cite{Reps98}]
  \creach{CFL} is in \P.
\label{prop:CFLinP}
\end{proposition}
\begin{proof}{(Sketch)}
  The proof is a dynamic programming algorithm. The algorithm maintains
  for each pair of vertices $u$ and $v$ a table entry $Y[u, v]$ such
  that $Y[u, v]$ is the set of all non-terminals $V$ in the grammar such
  that there is a path from $u$ to $v$ with yield that can be derived
  from $V$. The algorithm can be modified to output the derivation for
  $x$ where $x \in \lang{L}$ is the yield of a path from $s$ to
  $t$. Note that this implies that for all ``yes'' instances there
  exists a string with length of the derivation at most polynomial in
  the size of the graph.
\end{proof}

Sudborough \cite{Sudborough78} studied the class of languages
logspace reducible to a \class{CFL}. This class is called
$\LogCFL$. Sudborough \cite{Sudborough78} also showed that $\LogCFL$\ can
be characterized as the set of all languages accepted by an
\nauxpdapoly. An \nauxpdapoly\ is an NTM with a read-only input tape
and a logspace read-write work tape. It also has a pushdown stack
available for auxiliary storage. The machine is allowed to run only
for a polynomial number (in the input length) of steps. It is also
known that the language \lang{NBC(D_2)} (Nondeterministic block choice
Dyck$_2$) is complete for the class $\LogCFL$. The language \lang{NBC(D_2)}
consists of all strings of the form
$x_1[x_2\#x_3][x_4\#x_5]\ldots[x_k\#x_{k+1}]$ where each $x_i$ is a
string of two types of parentheses. The string between ``['' and ``]''
is called a block and the symbol \# separates choices in a block. A string
is in the language \lang{NBC(D_2)} if and only if there is a choice of $x_i$'s from each block such
that the final string (after all choices have been made) is in $D_2$.

\section{Formal Language Class restricted Reachability}

We know that \creach{REG} is in \NL~\cite{Yanna90}. The algorithm works by
constructing the product automata of the input graph and the DFA for
the regular language. The problem then reduces to the reachability
problem on the product automata. One problem with this approach is
that even if the input graph is an undirected graph, the product
automata will be a directed graph. We know that reachability in
directed graphs is harder than reachability in undirected graphs. The
following theorem shows that for regular languages, restricted directed
and undirected reachability are equivalent.

\begin{theorem}
  If \lang{L} is the regular language \lang{L((ab)^{*})} over the alphabet $\{a,b\}$
  then \lureach{L} is \NL-complete.
  \label{thm:regureach}
\end{theorem}

\begin{proof}
  To show that \lureach{L} is \NL-hard, we give a logspace reduction
  from \lang{REACH}. Given an instance $(G, s, t)$ of \lang{REACH} we
  construct an instance $(G', s, t)$ of \lureach{L} where $G'$
  is a labelled undirected graph where each edge is labelled either
  $a$ or $b$. The vertex set of $G'$ is given by $V(G') = V(G)
  \cup \{ m_{uv} : (u, v) \in E(G) \}$. For each edge $(u, v) \in
  E(G)$, we add two undirected edges $\{u, m_{uv}\}$ labelled $a$ and
  $\{m_{uv}, v\}$ labelled $b$ to $E(G')$. It is easy to see that
  any directed path from $s$ to $t$ corresponds to a path from $s$ to
  $t$ in $G'$ labelled by a string in \lang{L} and vice versa.
\end{proof}

So we know that \creach{REG} is \NL-complete and \creach{CFL} is at
least as hard as $\LogCFL$. So it is interesting to consider the
complexity of \creach{LIN}. We know that $\class{REG} \subseteq$
$\class{LIN} \subseteq \class{CFL}$ in the formal language theory
setting. The following theorem shows that \creach{LIN} is equivalent
to \creach{REG}.

\begin{theorem}
  \ctreach{LIN}, \cdreach{LIN}, \cureach{LIN} and \creach{LIN} are all
  $\NL$-complete.
  \label{thm:lin}
\end{theorem}

\begin{proof}
  There is an \NL-complete language in \class{LIN} \cite{Sudb75}. The
  hardness follows from this fact and Observation~\ref{obs:L-to-Lreach}. Now we show that all these problems
  are in \NL. The Dynamic Programming algorithm for \creach{CFL} %\ssay{CFL-REACH??} 
  from Proposition~\ref{prop:CFLinP} runs in poly-time and
  produces a polynomial length derivation for the output string
  (string yielded by the path). For any language in \class{LIN}, a
  polynomial length derivation can only produce a polynomial length
  string (and hence polynomial length path). Let us say that the
  length of the path is bounded by $n^{k}$ where $n$ is the size of
  the graph and $k$ is a constant. Then our algorithm will search for
  a path of length at most $n^{k}$ by nondeterministically guessing
  the next vertex at each step and simultaneously parsing the string
  at each step (using a 1-turn \class{PDA}.). This can be implemented
  by a 1-turn \class{AuxPDA} that runs in time $n^{k}$ and takes
  $\log(n)$ space. Sudborough \cite{Sudb75} proved that this class is exactly
  the same as \NL.
\end{proof}

The following theorem shows that for solving reachability for
\class{DCFL}s (which are nondeterministic), some nondeterminism is
unavoidable.

\begin{theorem}
  \cdreach{DCFL} is $\LogCFL$-complete.
  \label{thm:dcfl-reach}
\end{theorem}
\begin{proof} Let $\lang{L} \in \class{DCFL}$. We will describe an
  \nauxpdapoly\ that decides the language \ldreach{L}. The machine
  starts with the source vertex $s$ as the current vertex. At each
  step it nondeterministically moves to an out-neighbor of the
  current vertex. When the machine takes the edge $(u, v)$ it executes
  one step of the DPDA for \lang{L}, using the stack and finite
  control, with the label on $(u, v)$ as the current input symbol.
  The machine accepts iff it reaches $t$ and the DPDA accepts.

  For hardness, we reduce \lang{NBC(D_2)} to \cdreach{DCFL}. The
  reduction results in a series-parallel graph as shown in
  Figure~\ref{fig:dcfldagreach}. In the figure, a dashed arrow
  represents a simple path labelled by the given string. Note that the
  language \lang{D_2} is in \class{DCFL}.
\end{proof}
%\vspace{-5mm}
\begin{figure}
\centering
\includegraphics[scale=0.75]{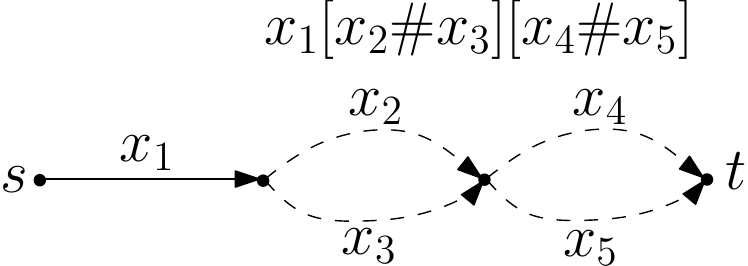}
\caption{Reducing \lang{NBC(D_2)} to \cdreach{DCFL}}
\label{fig:dcfldagreach}
\end{figure}
%\vspace{-5mm}
\begin{proposition}
\ctreach{CFL} and \cdreach{CFL} are $\LogCFL$-complete.
\label{prop:cfl-tdreach}
\end{proposition}
\begin{proof}
Sudborough \cite{Sudb75} defines a context-free language that is
complete for the class $\LogCFL$. This shows the hardness. To show
membership in $\LogCFL$\ consider an \nauxpdapoly\ that starts with $s$
as the current vertex and at each step guesses the next vertex while
simultaneously using the stack to simulate the parsing of the
CFL. This machine accepts iff the current vertex is $t$ at some point
and the PDA is in an accepting state at the same time. It
is easy to see that this \nauxpdapoly\ decides these languages.
\end{proof}

We now give a simplified presentation of a known result that says that 
\creach{CFL} is \P-complete. Also observe that 

\begin{theorem}[\cite{UllmanG88}]
  Let \lang{D_2}  ($\epsilon$-free Dyck$_2$) be the \class{CFL} given by the grammar 
  \begin{equation*}
    S \rightarrow (S)\ |\ [S]\ |\ SS\ |\ (\ )\ |\ [\ ].
  \end{equation*}
  \lreach{D_2} is \P-complete.
  \label{thm:cfl-reach-pc}
\end{theorem}

\begin{proof}
  This theorem has been proved in \cite{UllmanG88} using a different
  terminology. Here we give a simplified presentation of the proof
  using our terminology for the hardness of this language.  We show
  the \P\ hardness for \lang{D_2} by reducing a \P-complete problem MCVP\ (Monotone Circuit Value
  Problem where fan-out and fan-in of each gate is at most 2) to
  \lreach{D_2}. We may assume without loss of generality that each gate in the input circuit
  has fan-out at most 2. The reduction works by replacing each gate by
  a gadget as shown in Figure~\ref{fig:d2reach}. Each gadget in the
  construction has an input vertex and an output vertex.  The gadgets
  for input gates are straightforward. For an AND gate we add 3 new
  vertices and connect them to the gadgets for two gates feeding input
  to the AND gate. Suppose that the left input to the AND gate comes
  from the $2^{nd}$ ($1^{st}$) output wire of the left input
  gate. Then the first and second edges are labelled by ``$[$''
    (``$($'' resp.) and ``$]$'' (``$)$'' resp.) respectively.

  We use proof by induction on the level of the output gate of the
  circuit to prove the correctness of this reduction.  The inductive
  hypothesis is that there is a valid path from the input vertex to
  the output vertex of a gadget iff the output of the gate is 1 and
  any path that enters a gadget through its input gate and leaves it
  from some vertex other than its output vertex will be invalid.  This
  holds trivially for gadgets for the input gates.  Now any valid path
  from the input vertex to the output vertex of the AND gadget must
  consist of valid subpaths within the gadgets for the gates feeding
  input to this AND gate.  The only exception is when some path leaves
  this gadget for the AND gate from some vertex other than its output
  vertex. Note that by the induction hypothesis such a path can only
  leave from vertex $w$ or $z$ of the gadget. But the vertex $w$ (also $z$)
  has out-degree at most 2 and the other edge will be labelled by a
  closing bracket that does not match the type of bracket on the edge
  $(u, v)$.  This mismatch invalidates the path.  A similar argument
  holds for OR gates. This completes the induction.
\end{proof}

\begin{figure}[ht!]
  \centering
  \includegraphics[scale=0.5]{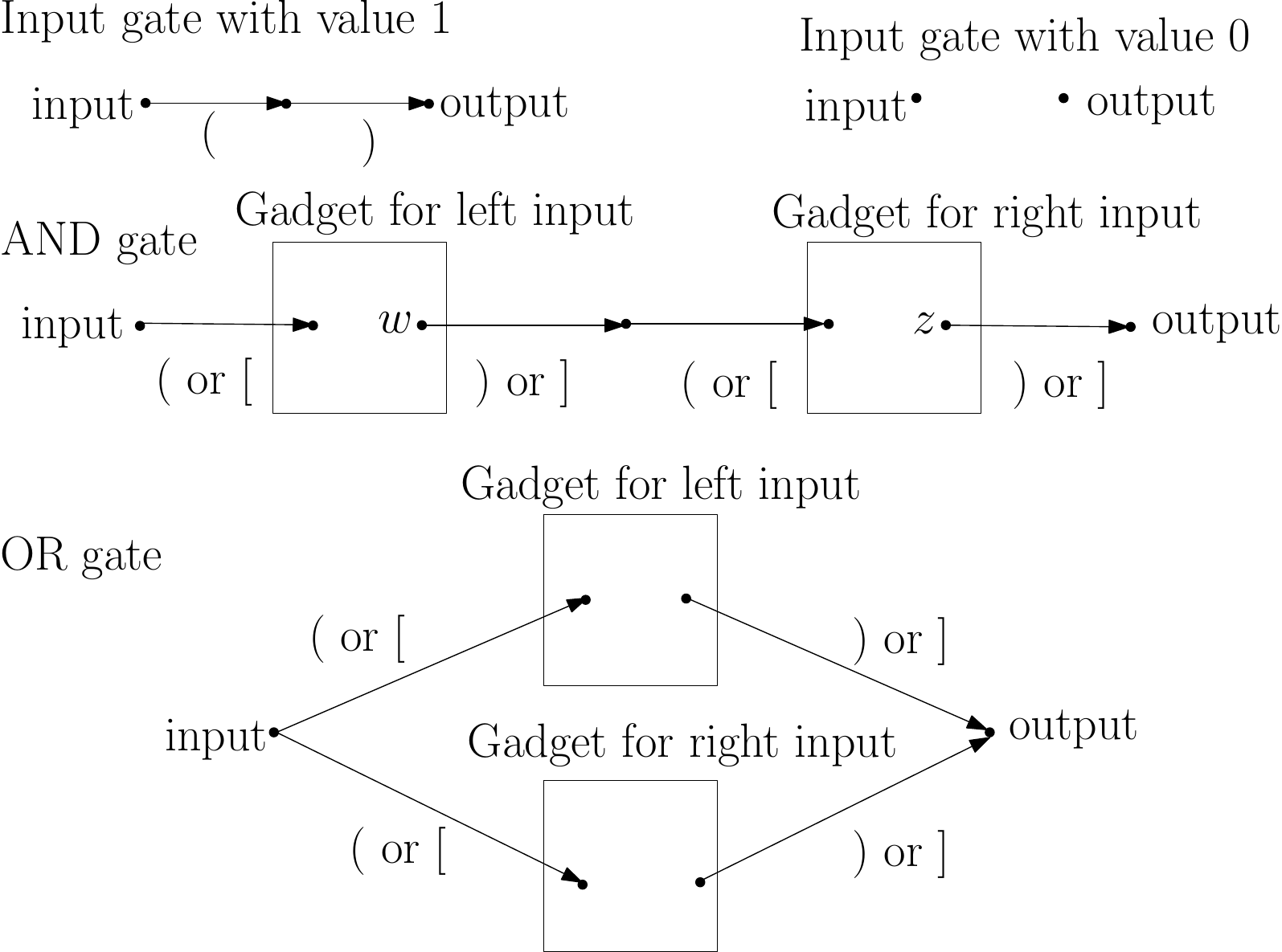}
  \caption{Reducing MCVP\ to \lreach{D_2}}
  \label{fig:d2reach}
\end{figure}

Now we prove a theorem similar in spirit to
Theorem~\ref{thm:regureach} for \class{CFL}s.  The proof uses
the same idea to make the undirected version as hard as the directed
one.

\begin{theorem}
  Let \lang{DD_2} be the CFL given by the grammar 
  \begin{equation*}
    S \rightarrow (aSb)\ |\ [cSd]\ |\ SS\ |\ (a\ b)\ |\ [c\ d].
  \end{equation*}
  \lureach{DD_2} is \P-complete.
  \label{thm:cflureach}
\end{theorem}
\begin{proof}
  \cureach{CFL} is in \P\ by \cite{UllmanG88}. We prove hardness by
  reducing from \lreach{D_2}. The reduction works by replacing each
  edge of the \lreach{D_2} instance by an undirected path of length
  two.  If for two vertices $a,b$, the directed edge from $a$ to $b$ is
  labelled ``$($'' (respectively ``$)$'',``$[$'' and ``$]$'') then
  replace it by an undirected path of length two with yield
  ``$(a$''(respectively ``$b)$'',``$[c$'' and ``$d]$'') when read from vertex $a$ to vertex $b$.
 %as shown in Figure~\ref{fig:dd2-ureach.jpeg}. 
The correctness of the reduction is easy to see.
\end{proof}

\noindent We state the following proposition, which follows from Theorem~\ref{thm:pspace-np-inv}.
\begin{proposition}
\ctreach{CSL} and \cdreach{CSL} are \PSPACE-complete.
\label{prop:csl-tdreach}
\end{proposition}

\section{Complexity Class restricted Reachability}

Now we consider the complexity of \lreach{L} and its variants when
\lang{L} is chosen from complexity classes. Barrett
et. al. \cite{Marathe00} has shown that even for languages in \L, the
languages \lreach{L} and \lureach{L} are undecidable. But note that
for any decidable \lang{L}, the language \ldreach{L} is decidable. So
we restrict our study only to \ldreach{L} in this section.

We have seen that moving up in the Chomsky hierarchy increases the
complexity of reachability. It is natural to ask whether such an
observation also holds with respect to the complexity classes, i.e., increasing the
complexity of \lang{L} increases the complexity of \ldreach{L}. More
concretely, does $\lang{A} \lmreduces \lang{B}$ imply 
$\ldreach{A} \lmreduces \ldreach{B}$.  The following theorem (which we restate from the introduction) shows
that this is very unlikely.

\vspace{3mm}

\begin{theorem-a}{\ref{thm:main}}
 There is an $\lang{A} \in \TC^0$ for which $\ldreach{A}$ is
  \NP-complete.
\end{theorem-a}
\begin{proof}
 The language \lang{A} can be thought of as an encoding of vertex cover.
 Each string $w$ in \lang{A} consists of 3 parts, say $w_1, w_2$ and
 $w_3$.  $w_1$ is a string of the form $1^k0^{n-k}$ and encodes $k$,
 the size of vertex cover, in unary. $w_2$ consists $n\choose 2$ bits
 which is the adjacency matrix representation of the input graph.
 $w_3$ consists $n$ bits which encodes the vertex cover by the
 characteristic vector. The strings $w_1, w_2$ and $w_3$ are separated
 by a $\#$ and each of the $n$ bits in $w_3$ is separated by a $\#$.

Let $n_1(x)$ be the number of 1's in the string $x$. A string $w$ is
in the language \lang{A} iff the following conditions hold.
\begin{enumerate}
\item The size of the vertex cover must be at most the size given in the first
  part of $w$.\\ ie., $n_1(w_3) \leq n_1(w_1)$, and

\item If the edge $\{i,j\}$ is present in the graph, then either the
  $i^{th}$ or the $j^{th}$ vertex must be present in the vertex
  cover. \\ ie., $(w_2(i,j)=1)\implies ((w_3(i)=1)\lor(w_3(j)=1)$.
\end{enumerate}

Any string  $w\in\lang{A}$ can be expressed as 

\begin{eqnarray*}
(\forall_{i,j}(w_2(i,j)=1) & \implies & (w_3(i)=1\lor w_3(j)=1))\land\\
  & &\exists k\le n, ((n_1(w_1)=k)\land (n_1(w_3)\leq k))
\end{eqnarray*}

% where $x$ is the size given in the first part and $z$ is the size of vertex cover in the third part.

An $\AC^0$ circuit is enough to check the conditions
$(\forall_{i,j}(w_2(i,j)=1) \implies (w_3(i)=1\lor w_3(j)=1))$ and
$\exists k\le n, (n_1(w_1)=k)$ but a $\TC^0$ circuit is necessary to
check whether $n_1(w_3)\leq k$.

A sketch of the structure of the circuit is given in Fig~\ref{tc0-circuit}.

\begin{figure}[ht!]
\centering \includegraphics[scale=0.84]{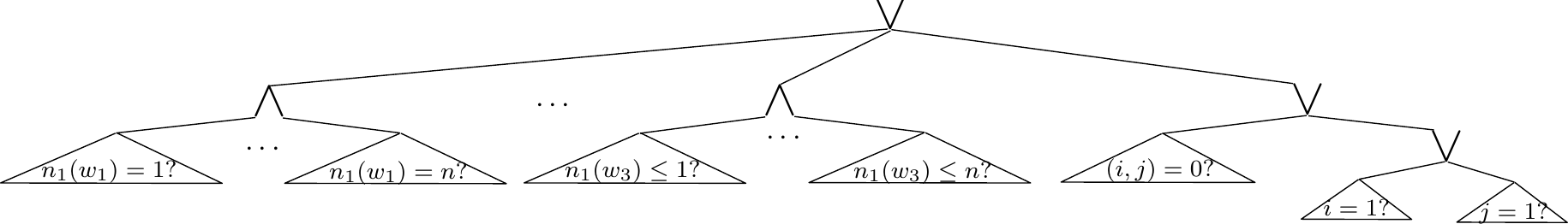}
\caption{Circuit for $\lang{A}$}
\label{tc0-circuit}
\end{figure}

To show \NP-hardness, we reduce \lang{VERTEX\lower-.12em\hbox{--}COVER}
to \ldreach{A}.
 
The language $\ldreach{A}$ is in $\NP$ as the non-deterministic Turing
machine guesses the path and verifies whether the yield of the path is
in \lang{A}.

 The reduction is given in Fig~\ref{AREACHisNPC1}. The DAG contains
 three parts. The first part, path from $s$ to $t_1$ encodes the size
 of the vertex cover and the second part, from $t_1$ to $t_2$ encodes
 the graph while the third part, from $t_2$ to $t$ represents the
 actual vertex cover.

\begin{figure}[ht!]
\centering \includegraphics[scale=0.70]{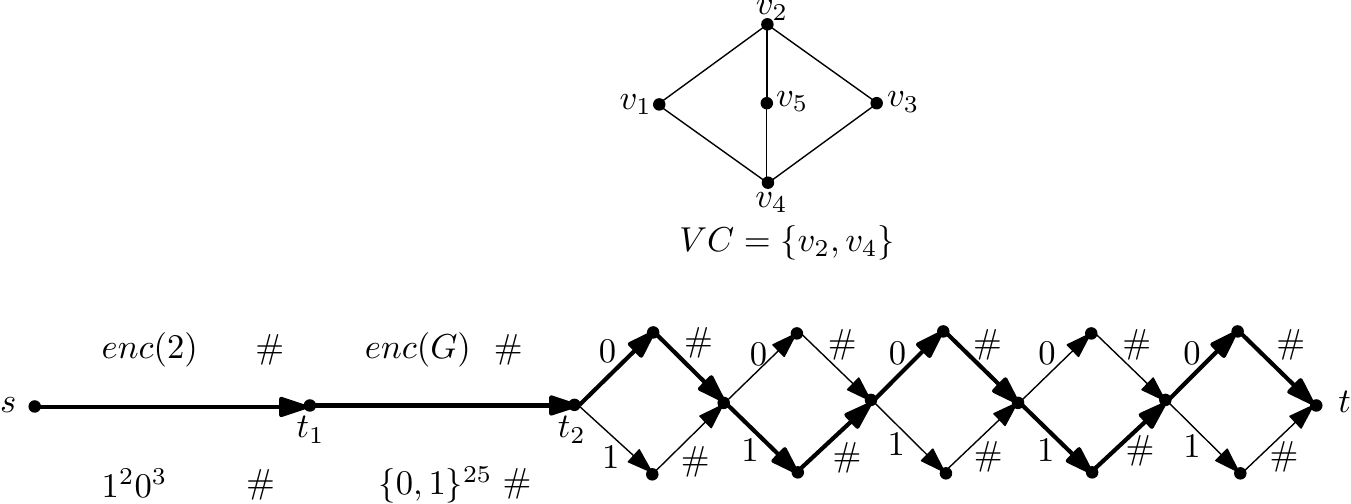}
\caption{Reducing VERTEX COVER to $\ldreach{A}$}
\label{AREACHisNPC1}
\end{figure}

For every $w\in \lang{A}$ we construct a valid path in DAG as
follows. Take the edges labelled by $1$, corresponding to the $1's$ in
the third part of $w$ (it is same as the vertices in the vertex cover).  
For the remaining vertices, the edges labelled
$0$ will be taken in the path.

Every valid path in the DAG corresponds to a vertex cover in $G$.  Let
$w$ be the yield of the path and let $w_3$ be its third part. Then
include $i$ in the vertex cover iff the $(2i - 1)^{\hbox{th}}$ symbol
of $w_3$ is 1.
\end{proof}

We are now going to see how the above result can be used for
translating separations of complexity classes upwards (Theorem~\ref{thm:pdr}).
For any complexity class $\class{C}$, we consider the class of
languages defined as $\dagr{C} = \{ \lang{L} : \ldreach{L} \in
\class{C}\}$. We have the following theorems for natural choices of
$\class{C}$. Note that for any class $\class{C}$, we have
$\dagr{C} \subseteq \class{C}$.

\begin{theorem}
  $\dagr{\PSPACE} = \PSPACE$ and $\dagr{\NP} = \NP$.
  \label{thm:pspace-np-inv}
\end{theorem}
\begin{proof}
  Let $\lang{L} \in \PSPACE$, then given an instance of $\ldreach{L}$
  we enumerate all paths from $s$ to $t$ and run
  the \PSPACE\ algorithm for $\lang{L}$ on the yield.  This is
  a \PSPACE\ algorithm for $\ldreach{L}$.
  Similarly if $\lang{L} \in \NP$, then a path from $s$ to $t$ along
  with the certificate for the yield on that path is a poly-time
  verifiable certificate for the $\ldreach{L}$ problem.
\end{proof}

\begin{theorem}
  \label{thm:pdr}
  $\P \neq \dagr{\P} \iff \P \neq \NP$.
  \label{thm:p-np}
\end{theorem}
\begin{proof}
  Suppose $\P \neq \dagr{\P}$ and let $\lang{L} \in \P \setminus
  \dagr{\P}$. Now $\ldreach{L}$ is in \NP\ by Theorem~\ref{thm:pspace-np-inv}.  By
  the choice of $\lang{L}$ we also have $\ldreach{L}$ is not in $\P$.

  For the other direction: suppose $\dagr{\P} = \P$.  We know that
  there is a language $\lang{L} \in \P$ for which \ldreach{L} is
  \NP-complete. Hence, $\P = \NP$.
\end{proof}

Theorem~\ref{thm:p-np} shows that separating $\dagr{\P}$ from $\P$ would
separate $\P$ from $\NP$.  This gives us an upward translation of
lower bounds on complexity classes provided we can prove that 
\dagr{\P} is contained in some subclass of \P. The interesting question is
whether we can identify some ``natural'' complexity class containing
$\dagr{\P}$.

By using similar arguments, we also have

\begin{theorem}
\label{thm:np-nl}
  $\dagr{\NL} \neq \NL \iff \NP \neq \NL$.
\end{theorem}

However \dagr{\NL} contains \NL-complete languages (See
Theorem~\ref{thm:lin}). So proving that \dagr{\NL} is separate from
\NL\ could be very hard.

The following theorem can be viewed as an evidence that \dagr{P} could be
separate from \P. A language \lang{L} is syntactic read-once logspace
(this notion was considered by Hartmanis et. al. in \cite{hartmanis})
reducible to another language \lang{L'} iff there is a logspace
reduction from \lang{L} to \lang{L'} and in the configuration graph
for this reduction all paths from the start configuration to the
accepting configuration reads each input variable at most once. It
shows that if we restrict our attention to syntactic read-once
logspace reductions, then \ldreach{L} for a \P-complete problem \lang{L}
is \NP-complete. Note that many natural \P-complete problems such as
\CVP\ (Circuit Value Problem) remains \P-complete even under syntactic
read-once logspace reductions.\\[1mm]

\begin{theorem-a}{\ref{thm:rol}}
  If $\lang{L}$ is $\P$-complete under syntactic read-once logspace
  reductions, then $\ldreach{L}$ is $\NP$-complete.
\end{theorem-a}
\begin{proof}
  Let $\lang{V} \in \NP$ via a poly-time verifier $\mach{N}$.  Let
  $\lang{W}$ be the witness language for $\lang{V}$. i.e., $\lang{W} =
  \{ (x, w) : N(x, w) = 1 \text{ and } |w| = {|x|}^{k} \text{ for some
  } k\}$. Since \lang{L} is \P-complete $\lang{W}$ is read-once logspace 
  reducible to $\lang{L}$ via $\mach{M}$.  We reduce $\lang{V}$
  to $\ldreach{L}$.  Let $x$ be our input.  Take the configuration
  graph $G$ of $\mach{M}$ on length $|x| + {|x|}^{k}$ inputs (after
  fixing the value of $x$) and label each edge by the symbol output by
  the machine $\mach{M}$ in that step.  This graph $H$ is considered
  as an input to the language $\ldreach{L}$. First we prove that $H
  \in \ldreach{L}$ implies that $x \in \lang{V}$. Consider a path from
  $s$ to $t$ in $H$ labelled by a string in \lang{L}. This path
  corresponds to a witness string for $x$. Therefore there exists a
  string $w$ for which $(x, w)$ in $\lang{W}$ which implies $x \in
  \lang{V}$. For the other direction let $x \in \lang{V}$. Therefore
  there exists a string $w$ such that $(x, w) \in \lang{W}$. Now take
  the path in $G$ that corresponds to this $w$. The yield of this path
  is a member of the language \lang{L} since \mach{M} outputs this
  yield when given $(x, w)$ as input.
\end{proof}

\section{Discussion and Open Problems}

The main result of our work is the observation that if we can prove
that the class \dagr{\P} is contained in some complexity class that is
a subclass of \P, then we can translate separation of complexity
classes upwards. We propose the following open problem. \\[3mm]
\noindent{\bf Open Problem 1: } Prove that $\dagr{\P} \subseteq \NC$.\\

It would be interesting to study the behavior of \dagr{.} operator on
complexity classes below \NL. $\AC^0$ is the class of all languages
computable by poly-size, constant depth uniform Boolean circuits. Can
we say anything about the set of languages \dagr{\AC^0}? The only
languages \lang{L} for which we know that \ldreach{L} is in $\AC^0$
are finite languages. Recall that \class{DAGREACH} is \NL-complete and
we know that $\NL \neq \AC^0$. Therefore, for any language \lang{L}
such that \ldreach{L} is in $\AC^0$, the \ldreach{L} problem is strictly
easier than \lang{DAGREACH}. This leads us to our second open
problem.\\[2mm]
\noindent {\bf Open Problem 2:} Prove that if $\ldreach{L} \in \AC^0$ then $\lang{L} \text{ is finite}$.

\bibliographystyle{fundam}
	\bibliography{ref}

\begin{thebibliography}{10}

\bibitem{arora-barak}
Arora, S., Barak, B.: \emph{Computational Complexity: A Modern Approach},
\newblock Cambridge University Press, 2009,
\newblock ISBN 9780521424264.

\bibitem{Marathe00}
Barrett, C.~L., Jacob, R., Marathe, M.~V.: Formal-Language-Constrained Path
  Problems,
\newblock \emph{SIAM Journal of Computing}, \textbf{30}(3), 2000, 809--837.

\bibitem{Barrington89}
Barrington, D. A.~M.: Bounded-Width Polynomial-Size Branching Programs
  Recognize Exactly Those Languages in {NC$^1$},
\newblock \emph{Journal of Computer and System Sciences}, \textbf{38}(1), 1989,
  150--164.

\bibitem{Barrington98searchingconstant}
Barrington, D. A.~M., Lu, C.-J., Miltersen, P.~B., Skyum, S.: Searching
  constant width mazes captures the {AC$^0$} hierarchy,
\newblock \emph{In Proceedings of the 15th Annual Symposium on Theoretical
  Aspects of Computer Science}, Springer-Verlag, 1998.

\bibitem{Hansen04}
Hansen, K.~A.: Constant Width Planar Computation Characterizes {ACC$^{0}$},
\newblock \emph{Proceedings of the 21st Annual Symposium on Theoretical Aspects
  of Computer Science}, 2004.

\bibitem{hartmanis}
Hartmanis, J., Immerman, N., Mahaney, S.~R.: One-Way Log-Tape Reductions,
\newblock \emph{Proceedings of 19th Annual Symposium on Foundations of Computer
  Science}, 1978.

\bibitem{hopcroft-ullman}
Hopcroft, J.~E., Motwani, R., Ullman, J.~D.: \emph{Introduction to automata
  theory, languages, and computation - international edition (2. ed)},
\newblock Addison-Wesley, 2003,
\newblock ISBN 978-0-321-21029-6.

\bibitem{Reps90}
Horwitz, S., Reps, T.~W., Binkley, D.: Interprocedural Slicing Using Dependence
  Graphs,
\newblock \emph{ACM Transactions on Programming Languages and Systems},
  \textbf{12}(1), 1990, 26--60.

\bibitem{KomarathSS14}
Komarath, B., Sarma, J., Sunil, K.~S.: On the Complexity of L-reachability,
\newblock \emph{Descriptional Complexity of Formal Systems - 16th International
  Workshop, {DCFS} 2014, Turku, Finland, August 5-8, 2014. Proceedings}, 2014.

\bibitem{Reingold05}
Reingold, O.: Undirected connectivity in log-space,
\newblock \emph{Journal of the ACM}, \textbf{55}(4), 2008.

\bibitem{Reps96}
Reps, T.~W.: On the Sequential Nature of Interprocedural Program-Analysis
  Problems,
\newblock \emph{Acta Informatica}, \textbf{33}(8), 1996, 739--757.

\bibitem{Reps98}
Reps, T.~W.: Program analysis via graph reachability,
\newblock \emph{Information {\&} Software Technology}, \textbf{40}(11-12),
  1998, 701--726.

\bibitem{Sudb75}
Sudborough, I.~H.: A Note on Tape-Bounded Complexity Classes and Linear
  Context-Free languages,
\newblock \emph{Journal of the ACM}, \textbf{22}(4), 1975, 499--500.

\bibitem{Sudborough78}
Sudborough, I.~H.: On the Tape Complexity of Deterministic Context-Free
  Languages,
\newblock \emph{Journal of the ACM}, \textbf{25}(3), 1978, 405--414.

\bibitem{UllmanG88}
Ullman, J.~D., van Gelder, A.: Parallel Complexity of Logical Query Programs,
\newblock \emph{Algorithmica}, \textbf{3}, 1988, 5--42.

\bibitem{Yanna90}
Yannakakis, M.: Graph-Theoretic Methods in Database Theory,
\newblock \emph{Proceedings of the 9th ACM Symposium on Principles of Database
  Systems}, 1990.

\end{thebibliography}

\end{document}